\documentclass[twoside,leqno,twocolumn]{article}
\pdfoutput=1
\usepackage[letterpaper]{geometry}
\usepackage{ltexpprt}
\usepackage[unicode]{hyperref}
\usepackage{algorithm}
\usepackage{algpseudocode}
\usepackage{amsmath,amssymb,amsfonts}
\usepackage{array}
\usepackage{booktabs} 
\usepackage[numbers]{natbib}
\usepackage{graphicx}
\usepackage{mathtools}
\usepackage{multirow}
\usepackage{subcaption}
\usepackage{textcomp}
\usepackage[flushleft]{threeparttable}
\usepackage{xcolor}
\usepackage{xspace}

\def\er{Erd\H{o}s-R\'{e}nyi}
\def\X{\mathcal{X}}
\def\W{\mathcal{W}}

\def\M{\mathcal{M}}

\def\x{\mathrm{x}}
\def\eps{\varepsilon}
\DeclareMathOperator*{\argmax}{arg\,max}

\def\etal{\emph{et~al.}\xspace}
\newtheorem{definition}{Definition}

\newcommand{\denselist}{\itemsep 0pt\parsep=1pt\partopsep 0pt}
\newcommand{\bitem}{\begin{itemize}\denselist}
\newcommand{\eitem}{\end{itemize}}
\newcommand{\benum}{\begin{enumerate}\denselist}
\newcommand{\eenum}{\end{enumerate}}

\begin{document}
\title{Influencers and
the Giant Component: the Fundamental Hardness in Privacy Protection for Socially Contagious Attributes}

\author{Aria Rezaei\\
arezaei@cs.stonybrook.edu\\ Stony Brook University
\and
Jie Gao \\ jg1555@rutgers.edu\\Rutgers University
\and
Anand D. Sarwate \\ anand.sarwate@rutgers.edu\\Rutgers University}
\date{}
\maketitle


\begin{abstract}
\small\baselineskip=9pt
The presence of correlation is known to make privacy protection more difficult. We investigate the privacy of \emph{socially contagious} attributes on a network of individuals, where each individual possessing that attribute may influence a number of others into adopting it. We show that for contagions following the Independent Cascade model there exists a giant connected component of infected nodes, containing a constant fraction of all the nodes who all receive the contagion from the same set of sources. We further show that it is extremely hard to hide the existence of this giant connected component if we want to obtain an estimate of the activated users at an acceptable level. Moreover, an adversary possessing this knowledge can predict the real status (``active'' or ``inactive'') with decent probability for many of the individuals \emph{regardless} of the privacy (perturbation) mechanism used. As a case study, we show that the Wasserstein mechanism, a state-of-the-art privacy mechanism designed specifically for correlated data, introduces a noise with magnitude of order $\Omega(n)$ in the count estimation in our setting. We provide theoretical guarantees for two classes of random networks: \er~graphs and Chung-Lu power-law graphs under the Independent Cascade model. Experiments demonstrate that a giant connected component of infected nodes can and does appear in real-world networks and that a simple inference attack can reveal the status of a good fraction of nodes.
\end{abstract}
\section{Introduction}\label{sec:intro}

Protecting the privacy of sensitive user data has always been a major challenge in the big data regime. Without it, users might be reluctant to share honest statistics with service providers and the integrity of collected data can suffer as a consequence. Recent findings have shown that the existence of correlation in data makes privacy protection an even more challenging task. Here, we focus at a special class of correlated data: attributes that spread through a network of individuals by jumping from one person to another. On top of contagious diseases that spread within communities, many behavioral attributes, such as political affiliation, obesity, or smoking habits, are also considered to be \emph{socially contagious}~\cite{christakis2007spread}. For an example of the privacy concerns, consider the following scenario: 
the local Department of Health (DoH) would like to know the fraction of residents who are regular smokers, by conducting a survey from the residents. To protect the privacy of the residents, each resident would report with a probabilistic perturbation (e.g., with a chance of $1/2$ report a random answer).  The DoH collects the perturbed reports and calculates the aggregate fraction by removing the introduced bias. 
If an adversary has the knowledge of the social ties
between the participants as well as the social influence model on smoking,
how confidently can she guess whether a given individual is a smoker or not? This can be phrased as a statistical inference attack. 

It is already recognized that Differential Privacy (DP), arguably the most commonly used privacy framework, along with other conventional privacy frameworks, claim no protection against additional outside information, including knowledge about correlation among data records. 
In our prior work~\cite{rezaei19privacy},
we have shown using experiments that there is a positive correlation between \emph{out-degree} of a node, which represents the total influence that a node exerts over others, and its vulnerability to an inference attack. 

This paper improves upon prior work in two ways. First, we show that obtaining good utility of estimating the fraction of infected users is impossible for a broad class of privacy mechanisms formulated by the success of statistical inference attacks. Secondly, our prior work provided only empirical evidence suggesting a relationship between node's influence and the loss of privacy.
In this paper we provide a rigorous explanation of this relationship. To the best of our knowledge, we provide first analysis of the fundamental challenges that exist in data privacy when the attributes are socially contagious and the contagion follows one of the most widely used models, the Independent Cascade model.

\smallskip
\noindent\textbf{Our Contribution:}
At the core of privacy protection, there is a trade-off between privacy and utility, characterized by the magnitude of noise added to data or query answer.
We show that under reasonable assumptions, finding a good trade-off between utility and privacy for such individuals could be extremely difficult. Our contributions are summarized as follows:
\bitem
\item For {\er}  and Chung-Lu power-law networks \cite{Chung-Lu-Networks}, we show in Theorems~\ref{theo:er_influencer_privacy} and \ref{theo:cl_influencer_privacy} that under reasonable assumptions there exists one giant connected component, containing a constant fraction of all nodes, with either all of its nodes possessing the contagious attribute or none. 
We next show that if an adversary has access to published estimate of the number of individuals with the target attribute, she can learn about the status of this giant component. Using this, she can infer the status of individuals in the network. Those that are most likely to be in this giant component have great influence over the activation of others and are called ``influencers''. We show that the influencers are most vulnerable to such an inference attack.
\item We show in Theorem~\ref{theo:hiding_c^H_1} that any privacy-preserving mechanism, in order to increase \emph{uncertainty} about the status 
of the giant component, needs to introduce an enormous amount of noise (a fraction of the number of all nodes), rendering low data utility.
\item We show in Theorem~\ref{theo:wass_noise} that the Wasserstein mechanism~\cite{SongYC:17pufferfish}, a state-of-the-art mechanism that is designed specifically to hide sensitive information in correlated data, will always add noise with a magnitude equal to a constant fraction of the number of all nodes, thus providing weak utility in all cases. This is regardless of the privacy requirements of the application and hence an even stronger result.
\item Through experiments on $4$ real-world networks, we show that the giant component indeed appears in real networks. Moreover, a sizeable number of nodes are in fact very likely to be a part of it and as a result, an adversary can easily infer their real activation status.
\eitem
Remark that we go beyond the known results on the difficulties of correlated data privacy under conventional frameworks, such as DP. We make minimal assumptions about the privacy mechanism in question and show that preventing an adversary from inferring the true status of certain nodes is impossible even for very weak utility guarantees.
Our results are thus applicable to a wide range of privacy mechanisms, including DP.

\section{Related Work}\label{sec:related}
Protecting the privacy of correlated data has been a longstanding challenge in the literature. One of the most commonly used privacy frameworks is Differential Privacy (DP).
Under DP, entities protect a data record $x$ against an adversary when publishing sensitive data or answering queries, by bounding the probability of the aggregate statistics being the same, in a database with or without $x$. ~\cite{DworkMNS06}.
Nevertheless, when data is correlated, an adversary could still infer the status of $x$, by using data records oter than $x$
~\cite{kifer2014pufferfish,he2014blowfish,ZhuXLZ:15correlated,Kifer2011-pu} (notice that this does not violate the DP guarantee)
Relaxed definitions of DP have been proposed, such as group differential privacy~\cite{dwork2013algorithmic} or specific models for correlated/dependent data~\cite{GhoshK:17inferential,ZhuXLZ:15correlated,ZhaoZP:17ddp}. However, the proposed methods do not look at network data or contagion models as we do here, and instead focus on temporal correlation (e.g. location, time series). 
The most closely related relaxation is the Pufferfish framework~\cite{kifer2014pufferfish}, where privacy is specifically defined for given pairs of ``secrets''  which should be \emph{indistinguishable}, and a set of plausible priors 
which can contain known correlation structures in the data. Specialized Pufferfish mechanisms have been proposed for some problems~\cite{kifer2014pufferfish,YangSN:15bayesian,he2014blowfish}. For general Pufferfish problems, Song \etal~\cite{SongYC:17pufferfish} proposed the Wasserstein Mechanism
which in the end adds
noise to a query's answer, with a variance tuned to the assumed prior distribution's structure. We discuss the Wasserstein mechanism in depth in Section \ref{sec:hiding_superspreader}.

Our work is also different from prior work that consider the network structure to be private and the goal is to report network statistics such as degree distributions~\cite{Task2012-hi}. The most relevant to our work in the one by 
Song \etal~\cite{SongYC:17pufferfish} and Rezaei \etal~\cite{rezaei19privacy}. The former assumes correlated attributes in a social network where nodes of the same community have the same attributes; a rather simplistic model for correlation compared to the widely used contagion models such as Independent Cascade or Linear Threshold. Our work and the one in~\cite{rezaei19privacy} both consider contagion models, but the results here are more comprehensive, and have stronger theoretical guarantees.

\section{Preliminaries}\label{sec:prelim}


\subsection{Data Model}\label{sec:data}
Suppose that there are $n$ individuals, each having a binary attribute $\x_i \in \{0, 1\}$.
Suppose further that the individuals influence one another through a network, $G(V, E)$, where each node $v \in V$ represents an individual and the existence of an edge between two nodes $v, u \in V$ means that the corresponding individuals exert influence on each other.
We assume that the attribute $\{\x_i\}$ is \emph{contagious}, meaning that starting from a \emph{seed set} of nodes, $S$ ($\left| S \right| = s$), the attributes follow a contagion process and spread through the network.
The attributes are considered to be sensitive information.
In the following, we formally define the contagion process and present the privacy problem statement.
\subsection{Contagion Model}\label{sec:contagion_model}
One of the most commonly used models for contagions is the Independent Cascade (IC) model, used predominantly to model the spread of disease in communities.
Let $S$ denote the initial set of nodes that are activated before a contagion starts. Upon activation, each node has exactly one chance to activate each of its neighbors with probability equal to a \emph{transmission rate}, $q \in (0, 1]$. The cascade stops when no additional node is activated.
Similar to some of the previous works~\cite{JustFewRandomSeeds}, when an edge exists between nodes $v$ and $u$, we couple the event of $v$ activating $u$ to the event of $u$ activating $v$.

To analyze cascades over networks, Kempe et al.~\cite{kempe2003maximizing} have introduced the use of \emph{Triggering Set} technique. Suppose that at time $0$ we flip a biased coin with probability $q$ for each edge in the graph and choose to ``trigger'' the edge or not. After that, a node $v$ is activated if and only if there is a path between $v$ and one of the seed nodes via triggered edges. In the case of IC, it suffices for a seed node to appear in the connected component of $v$ for $v$ to become activated.
\subsection{Privacy and Utility}\label{sec:privacy_model}
We consider applications where the number of activated nodes in a population of $n$ individuals, $X_n = \sum \x_v$, is of interest, with  $\X_n = \{0,1,\ldots,n\}$ the possible values of $X_n$. We drop the $n$ when it is clear from context. Since individuals may be reluctant to report their true attributes, an estimate or approximation of $X$ is made using a privacy-preserving mechanism $\M$ which, given the sensitive value $X$, produces a \emph{perturbed} value $\M(X)$ to report.
We measure the utility of $\M$ by the amount of noise it adds to the real value, $\left| \M(X) - X \right|$. We associate high utility with relatively low error, and precisely \emph{sublinear asymptotic error}, defined below.
\begin{definition}[Sublinear Asymptotic Error]\label{def:strong_utility}
Consider a sequence of randomized privacy mappings $\M_n  \colon \X_n \mapsto \X_n$, where $n$ is the total number of individuals. Then $\M_n$ guarantees \emph{sublinear asymptotic} error, if with high probability\footnote{$E$ happens with \emph{high} probability if $\lim_{n \rightarrow \infty} P(E) \rightarrow 1$.} (w.h.p) over $\M$ and the graph generation process we have:
\begin{equation}
    \lim_{n \rightarrow \infty} \frac{\left| \M_n(X_n) - X_n \right|}{n} \rightarrow 0,
\end{equation}
for all values of $X_n \in \mathcal{X}_n$.
\end{definition}
Note that the above indicates that $|\M_n(X_n) - X_n| = o(n)$ for all values of $X_n$.
The association between high utility and sublinear asymptotic error is based on real-world applications and is common in the literature~\cite{dwork2013algorithmic}.
In many cases, such as when a survey is conducted from a subset of a large population of size $n$, the margin of error of the resulting estimate from the sample population will be $o(n)$, in other words \emph{negligible} in comparison to the population size, $n$. Hence, a privacy mechanism that adds noise to $X$ that is, w.h.p, at most $o(n)$ is considered to have high utility since it does not change the value of $X$ in any significant way.

\section{Giant Components and the Privacy of Influencers}\label{sec:high_degree_privacy}
In this section, we investigate the hardness of protecting the privacy of nodes whose activation causes a large number of other nodes' activation, or \emph{influencers}.
Specifically, we will show that over two classes of random networks, if the privacy mechanism guarantees sublinear asymptotic error (from Definition~\ref{def:strong_utility}), one can infer the status of the influencers with high confidence regardless of the details of the privacy mechanism and the amount of noise added.
The specific definition of influencers will depends on the characteristics of the graph. For the two random networks that we study, influencers are simply those with highest degrees. In both networks, and in real graphs, under the right circumstances a \emph{giant component} consisting of a constant fraction of all nodes appears in the triggering set of the network. If any of the seed nodes end up in this giant component, all nodes inside of it will be activated.
We show that the knowledge of the activation of such a cluster of nodes allows an adversary mount a successful inference attack.
\subsection{\er~Graphs}\label{sec:er_graphs}
An \emph{\er}~ graph is an undirected and simple graph $G(V,E)$ with $n$ nodes ($|V| = n$) where each possible edge $(v,u)$ exists with probability $p \in (0, 1)$. 

In the case where $np > 1$, the size of the largest connected component, known as the \emph{giant} component, is $\Theta(n)$ w.h.p.
More specifically, suppose that $C_i$ is the $i^{\text{th}}$ largest connected component and $|C_i|$ denotes its size. Then, if $np > 1$, $|C_1| \sim yn$ where $y$ is the solution to the following~\cite{TheProbabilisticMethod}:
\begin{equation}\label{eq:er_y}
    e^{-npy} = 1 - y.
\end{equation}
Furthermore, all other components are \emph{tiny}; for all $i \geq 2$, we have $|C_i| = O(\log{n})$. Finally, in case where $np$ is a constant above $1$, w.h.p the maximum degree is~\cite{frieze2016introduction}:
\begin{equation}\label{eq:er_max_degree}
    \Delta \approx \frac{\log{n}}{\log{\log{n}}}, 
\end{equation}

Note that the triggering set of an Independent Cascade with transmission rate $q$ over $G(n,p)$ is itself an \er~graph $G(n,pq)$. This means that if $npq > 1$, a connected component with a fraction of all nodes exists; either nodes inside the cluster are all activated, or none of them are.

We now present our main theorem on the hardness of protecting privacy for influencers.
We denote the triggering set created by the process outlined in Section~\ref{sec:contagion_model} by $H$ and its $i^{\text{th}}$ largest connected component by $C^H_i$.
\begin{theorem}\label{theo:er_influencer_privacy}
Suppose that the nodes in $G(n,p)$ are numbered in the decreasing order of their degree; with node $i$ having the $i$-th highest degree.
Suppose further that the number of seeds $s = o(n/\log^2{n})$, the privacy mechanism $\M$ with maximum error $e_{\M}$ guarantees sublinear asymptotic error, and $npq > 1$. Then, for any $\epsilon \in [0, 1)$, there exists a $k \geq 1$ such that:
\begin{align}
    \Pr[\x_i = 1 \mid \M(X) \geq |C^H_1| - e_{\M}] &\geq 1 - \epsilon \nonumber \\
    \Pr[\x_i = 0 \mid \M(X) \leq s|C^H_2| + e_{\M}] &\geq 1 - \epsilon,\label{eq:er_influencer_privacy}
\end{align}
for all $i \leq k$ and large enough $n$. 
\end{theorem}
Here, we have assumed that the
number of initial seeds is not an unrealistically high number, as is usually the case in real-world scenarios. For instance, organic cascades, such as a contagious disease, originate from a handful of initial infected individuals.
We have also assumed that each node, upon activation, activates \emph{on average} at least one other node ($npq > 1$).

Notice that the two conditions in \eqref{eq:er_influencer_privacy} correspond to when $C^H_1$ is activated and
$\M(X) \geq |C^H_1| - e_{\M}$,
or when $C^H_1$ is inactive and 
$\M(X) \leq e_{\M} + s|C^H_2|$. Since $npq > 1$, w.h.p $C^H_1 \sim yn$ where $y$ is the solution to \eqref{eq:er_y} when $p$ is substituted with $pq$. This means that the two conditions in \eqref{eq:er_influencer_privacy} are also exhaustive since the difference between the two values on the right side of the inequalities is:
\begin{equation*}
    |C^H_1| - s|C^H_2| - 2e_{\M} = \Theta(n).
\end{equation*}
This means that one could infer the activation of $C^H_1$ from the value of $\M(X)$ with with high confidence. We make this claim more rigorous in Section~\ref{sec:hiding_superspreader}.

As an example, consider the case when $n = 10{,}000$, $s = \log{n} \leq 10$, $\left| \M(X) - X \right| \leq \sqrt{n}$, and $y = 0.4$ ($pq \approx 1.28$). In the event that $C^H_1$ is not activated we will have $\M(X) \leq \log^2({10{,}000}) + 100 \leq 200$, while in the event of $C^H_1$'s activation we have $\M(X) \geq 3900$. Since the value of $\M(X)$ reveals the activation status of $C^H_1$, we can use this information to bound $\Pr[\x_i = \cdot]$ using the two lemmas below the proof of which you can find in Section~\ref{app:trig_approx_proof} of the Supplementary material.
\begin{lemma}\label{lem:P(x_v|M(X))_bound}
With a given $\M(X)$, for node $i$ we have:
\begin{align*}
    {}& (1) \; \Pr[\x_i = 1 \mid \M(X) \geq |C^H_1| - e_{\M}] \geq \Pr[i \in C^H_1] \\
    {}& (2) \; \Pr[\x_i = 0 \mid \M(X) \leq s|C^H_2| + e_{\M}] \geq \Pr[i \in C^H_1].
\end{align*}
\end{lemma}
\begin{lemma}\label{lem:trig_approx}
When $n$ is large enough, the probability that a node $i$ with a fixed degree $k$ in $H$ is in $C^H_1$ is approximated by:
\begin{equation}\label{eq:i_C^H_1}
    \Pr[i \notin C^H_1] \sim - \exp(-ky).
\end{equation}
\end{lemma}
The degree of any node $i$ in $H$, denoted by $d^H_i$, follows a binomial distribution characterized by $d$, its degree in $G$, and $q$, the transmission rate. Applying Chernoff bound for lower tail, we have:
\begin{equation}\label{eq:upper_bound_dq/2}
    \Pr\left[d^H_i \leq \frac{dq}{2}\right] \leq -\exp\left(-\frac{dq}{8}\right).
\end{equation}
Since the function $\exp(-ky)$ is decreasing in $k$, we have:
\begin{multline}
    \Pr[i \notin C^H_1] \leq \exp(0)\Pr\left[d^H_i \leq \frac{dq}{2}\right] + \\ \exp\left(-\frac{dqy}{2}\right)\Pr\left[d^H_i > \frac{dq}{2}\right].
\end{multline}
Combining this with \eqref{eq:upper_bound_dq/2}, we have:
\begin{equation}
    \Pr[i \notin C^H_1] \leq \exp\left(-\frac{dq}{8}\right) + \exp\left(-\frac{dqy}{2}\right).\label{eq:i_not_in_C^H_1_bound}
\end{equation}
Plugging in the above node $1$ with maximum degree, $\Delta \approx \log{n}/\log{\log{n}}$ from \eqref{eq:er_max_degree}, we can write for any $\epsilon > 0$ and an $n$ large enough:
\begin{equation}\label{eq:prob_limit_omega_n}
    \Pr[1 \in C^H_1] \geq 1 - \epsilon,
\end{equation}
This satisfies the conditions in \eqref{eq:er_influencer_privacy} according to Lemma~\ref{lem:P(x_v|M(X))_bound}.
Depending on the value of $\epsilon$, there will be other high-degree nodes for which the bound in \eqref{eq:i_not_in_C^H_1_bound} goes below $\epsilon$ since its value decreases exponentially in $d$. This shows that the $k$ in Theorem~\ref{theo:er_influencer_privacy} is always $\geq 1$ which concludes the proof of Theorem~\ref{theo:er_influencer_privacy}.
\subsection{Chung-Lu Power-Law Networks}\label{sec:power-law}
Despite being extensively studied, \er~graphs do not represent many phenomena observed in reality; in particular, the \emph{power-law} degree distribution that is common in many real-world networks~\cite{barabasi1999emergence}. To study such networks, we turn to \emph{Chung-Lu} networks \cite{Chung-Lu-Networks} with a degree distribution that follows a power-law distribution.
Let $\langle w_1, \ldots, w_n \rangle$ ($w_i \in \mathbb{R}_+$) be weights assigned to $n$ individuals and $\ell_n = \sum w_i$. A \emph{Chung-Lu} network is comprised of $n$ nodes, where an edge between two nodes $i$ and $j$ appears with probability:

\begin{equation}\label{eq:cl_edge}
P_{ij} = \min\left(1, \frac{w_i w_j}{\ell_n}\right).
\end{equation}
To generate a network that resembles a power-law network, we assign the weights using the following function:
\begin{equation}\label{eq:power_law_function}
    w_i = \left[1 - F\right]^{-1}(i/n), \;\; F(x) = 1 - (d/x)^b,
\end{equation}
where $b$ determines the \emph{scale} of the power-law distribution and $d$ is the minimum expected degree in the graph. You can see that the $i$-th ranked node has a weight corresponding to the $i/n$ percentile of the distribution. Finding the solution to \eqref{eq:power_law_function} yields:
\begin{equation}\label{eq:w_i}
w_i = d\left(\frac{n}{i}\right)^{\beta},\; \beta = \frac{1}{b}.
\end{equation}
Note that the triggering set of a Chung-Lu network with parameters $d$ and $b$ for an IC with parameter $q$ is itself a Chung-Lu network with parameters $dq$ and $b$.
It has been shown~\cite{JustFewRandomSeeds,HofstadRandomGraph} that if (1) $b \in (0, 2]$ or (2) $b > 2$ and $d > (b-1)(b-2)$, the generated graph will have a giant component with a constant fraction of nodes, $\alpha n$ ($0 < \alpha < 1$), and many small components each being at most $O(\log{n})$ in size. 

As before, we denote the main graph by $G$, its triggering set by $H$, and the $i$-th biggest connected component of $H$ by $C^H_i$.
\begin{theorem}\label{theo:cl_influencer_privacy}
Suppose that $s = o(n/\log^2{n})$, the privacy mechanism $\M$ with maximum error $e_{\M}$ guarantees sublinear asymptotic error, and either (1) $b \in (0, 2]$ or (2) $b > 2$ and $dq > (b - 1)(b - 2)$. Then, for any $\epsilon > 0$, there is a subset of nodes nodes $V' \subseteq V$ that grows arbitrarily large in size ($|V'| = \omega(1)$) and for any node $i \in V'$, we have:
\begin{align}
    \Pr[\x_i = 1 \mid \M(X) \geq |C^H_1| - e_{\M}] &\geq 1 - \epsilon \nonumber \\
    \Pr[\x_i = 0 \mid \M(X) \leq s|C^H_2| + e_{\M}] &\geq 1 - \epsilon,\label{eq:cl_influencer_privacy}
\end{align}
\end{theorem}
The proof of this theorem follows arguments similar to Theorem~\ref{theo:er_influencer_privacy} and is available in full in Section~\ref{app:cl_influence_proof} of the Supplementary material. The main takeaway is the following: depending on the value of $b$, there could be $3$ upper-bounds for the value of $i$, all $3$ of which can grow arbitrarily in size:
\begin{align}
    0 < b < 1 \colon & \; i = o\left(n^b\right), \\
    b = 1 \colon & \; i = o\left(\frac{n}{\log{n}}\right) \\
    b > 1 \colon & \; i =  o(n).
\end{align}
The $3$ equations above show that
for any $\epsilon$ in Theorem~\ref{theo:cl_influencer_privacy}, the subset $V' \subseteq V$ includes all nodes ranked $i \leq f(n)$ where $f$ can be any increasing function of $n$ that meets the corresponding asymptotic condition depending on $b$, for instance $f(n) = \log^2{n}$ in the case of $b = 1$.
\subsection{Notes on General Graphs}\label{sec:notes_general_graphs}
In general graphs too, if one retains each edge of a connected graph with probability $p$, for values of $p$ above a certain threshold\footnote{Define the second-order average degree $\tilde{d}=\sum_v d_v^2/(\sum_v d_v)$ where $d_v$ is the degree of $v$. If $p>(1+\eps)/\tilde{d}$, then almost surely, the percolated subgraph has a giant component.
Otherwise, if $p<(1-\eps)/\tilde{d}$, almost surely, there would be no
giant component.
} a \emph{giant component} appears with a constant fraction of all nodes and all the other components are of size $o(n)$~\cite{Chung2009-op}. This means that our findings are applicable to general graphs as well, however, the vulnerability of each node depends on its probability of appearance in that giant component, which depends on the structure of the graph. We show through experiments in Section~\ref{sec:exp} that for reasonable values of the transmission rate, $q$, a giant component appears in the triggering set and a sizeable number of all the nodes frequently appear in that component and are hence vulnerable to the inference attack introduced in this section.
\section{Hiding the Status of the Giant Connected Component}\label{sec:hiding_superspreader}
\def\tvd{\text{TVD}}
As we saw in Section~\ref{sec:high_degree_privacy}, when a privacy mechanism fails to hide the status of the giant component, $C^H_1$, an adversary can infer the real status of individual nodes by estimating the probability of their presence in $C^H_1$.
Here, we show that if there is to be \emph{any uncertainty} about the activation status of $C^H_1$, then the privacy mechanism \emph{cannot} provide sublinear asymptotic error.
\begin{figure}[!tb]
    \centering
    \includegraphics[width=0.95\columnwidth]{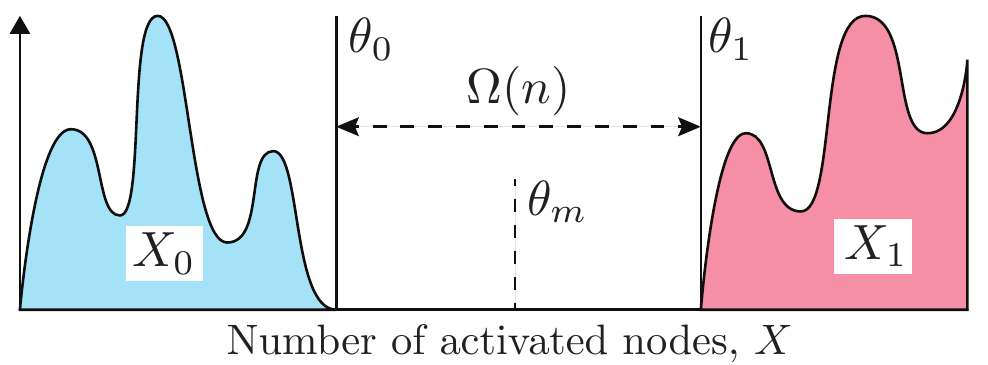}
    \caption{Probability distribution of $X$, when $C^H_1$ is \emph{not} activated, $X_0$ (left, blue), with maximum value $\theta_0$, and when $C^H_1$ \emph{is} activated, $X_1$ (right, red), with minimum value $\theta_1$. Note that $\theta_1 - \theta_0 = \Omega(n)$ in our setting.
    }
    \label{fig:X_0_and_X_1}
    \vspace*{-4mm}
\end{figure}
We denote by $X_1$ and $X_0$ the distribution of $X$ when $C^H_1$ is active and inactive, respectively. Let $\M \colon  [n] \mapsto [n]$ be a randomized privacy mechanism. Then, we denote the distribution of \emph{perturbed} values for $X_0$ and $X_1$ by $Z_0 = \M(X_0)$ and $Z_1 = \M(X_1)$. For brevity, we denote the probability density functions of these distributions  by the same symbol, i.e. $\Pr[\M(x \sim X_0) = a] = Z_0(a)$. We frame the problem as a hypothesis testing scenario. Suppose that the reported value $z$ is observed. Then, the two hypotheses are: (1) $H_0$, $z$ comes from $Z_0$ and (2) $H_1$, $z$ comes from $Z_1$.
\begin{theorem}\label{theo:hiding_c^H_1}
If the hypothesis testing error is at least a constant $0 < \epsilon < 1$, then $\M$ cannot provide sublinear asymptotic error.
\end{theorem}
\begin{proof}
Here, we make use of the relation between hypothesis testing error and the total variation distance ($\tvd$), calculated between two distributions $\mu$ and $\nu$ over a range $\X$ as below:
\begin{equation}
    \tvd(\mu, \nu) =
    \frac{1}{2}\sum_{x \in \mathcal{A}}\left| \mu(x) - \nu(x) \right|.
\end{equation}
Since the hypothesis testing error is $1 - \tvd(Z_0, Z_1)$ and given that it is at least $\epsilon$, for the total variation distance between $Z_0$ and $Z_1$ we have:
\begin{equation}\label{eq:tvd_upper_bound}
    \tvd(Z_0, Z_1) = \frac{1}{2}\sum_{a}\left| Z_1(a) - Z_0(a) \right| \leq 1 - \epsilon.
\end{equation}
Suppose that:
\begin{equation*}
    \theta_0 = \max X \sim X_0, \; \theta_1 = \min X \sim X_1,
\end{equation*}
and
$\theta_m = (\theta_0 +\theta_1)/2$.
Figure~\ref{fig:X_0_and_X_1} depicts the relation between the two distributions $X_0$ and $X_1$, and the values $\theta_0$, $\theta_1$ and $\theta_m$.
We can rewrite \eqref{eq:tvd_upper_bound} as:
\begin{align}
    \sum_{a \leq \theta_m} Z_1(a) - Z_0(a) + \sum_{a > \theta_m} Z_0(a) - Z_1(a) &\leq 1 - \epsilon \nonumber \\
    \sum_{a \leq \theta_m} Z_1(a) + \sum_{a > \theta_m} Z_0(a) &\geq \epsilon. \label{eq:sum_Z_above_epsilon}
\end{align}
This means that at least one of the terms on the left side
of \eqref{eq:sum_Z_above_epsilon}
is greater than or equal to $\epsilon/2$. Without loss of generality, we assume that it's $\sum Z_1(a)$. We then have:
\begin{equation}
    \sum_{a \leq \theta_m} Z_1(a) = \sum_{b}X_1(b) \Pr[\M(X) \leq \theta_m] \geq \frac{\epsilon}{2}.
\end{equation}
Assuming that $X^* = \argmax_{X \sim X_1}\Pr[\M(X) \leq \theta_m]$, since $\sum_{b} X_1(b) = 1$, we have:
\begin{equation}
    \Pr[\M(X^*) \leq \theta_m] \geq \frac{\epsilon}{2}.
\end{equation}
Since $\theta_1 - \theta_m = \Omega(n)$, we have at least one $X \sim X_1$ ($X^*$) for which \emph{with constant probability} the perturbed value, $\M(X^*)$, is $\Omega(n)$ distance away. This means that $\M$ cannot provide Sublinear Asymptotic Error.
\end{proof}
\subsection{Wasserstein Mechanism's Utility}\label{sec:wass_utility}
As a case study on the hardness of providing privacy in our setting, we turn to the Wasserstein mechanism~\cite{SongYC:17pufferfish}; a state-of-the-art privacy mechanism designed specifically for correlated data privacy. This mechanism is based on the Pufferfish framework~\cite{kifer2014pufferfish}. Similar to DP, Pufferfish tries to make pairs of \emph{secrets} indistinguishable up to a factor $e^\varepsilon$. For instance, here we need to make it hard for an adversary to distinguish between pairs of secrets: ``$\x_i = 1$'' and ``$\x_i = 0$'' where $i \in V' \subseteq V$.
A key difference between Pufferfish and DP is that in Pufferfish, these pairs of secrets can be \emph{non-exhaustive}, meaning that one can decide to protect only a subset of all individuals. A more detailed definition of the Pufferfish framework and its application in our setting is available in Section~\ref{app:pufferfish} of the Supplementary material.

Here, we show that in a setting similar to ours, for \emph{any} subset of individuals that the Wasserstein mechanism intends to protect, with reasonable utility requirements, the magnitude of the added noise added will be $\Omega(n)$. Note that this result is stronger than our previous findings for more general privacy mechanisms in Section~\ref{sec:high_degree_privacy}, as there is no inference attack involved and the mechanism adds $\Omega(n)$ noise to protect nodes with even little chance of being in the giant component.


The Wasserstein mechanism is
based on the $\infty$-Wasserstein distance between two distributions.
Suppose that $\mu$ and $\nu$ are two probability distributions on $\mathbb{R}$ and denote the set of all joint distributions with marginals $\mu$ and $\nu$ by $\Gamma(\mu, \nu)$. Then, the $\infty$-Wasserstein distance, or in short $W_\infty$, is calculated as below:
\begin{equation}\label{eq:w_inf}
    W_\infty(\mu, \nu) = \inf_{\gamma \in \Gamma(\mu, \nu)} \max_{x, y, \gamma(x, y) > 0} \left| x - y \right|.
\end{equation}
If you consider the probability mass of $\mu$ to be piles of dust along $\mathbb{R}$ and $\nu$ as the distribution of tiny holes with different depths, you can think of any $\gamma \in \Gamma(\mu,\nu)$ as a plan to move the piles of dust into the holes and $W_\infty$ as the minimum possible ``maximum distance any speck of dust travels from a pile in $\mu$ into a hole in $\nu$''.

Given $\mathbb{Q}$, the set of all pairs of secrets and
a sensitive query $F$, the Wasserstein Mechanism adds noise via $\text{Lap}(W/\varepsilon)$ to achieve $\varepsilon$-Pufferfish privacy where $W$ is:
\begin{equation}\label{eq:wasserstein_mechanism}
W = \sup_{(s_i, s_j) \in \mathbb{Q}} W_\infty\left(\mu_i, \mu_j\right),
\end{equation}
where $\mu_i = \Pr\left[F(X) = \cdot \mid s_i\right]$.
\citeauthor{SongYC:17pufferfish} have proven that the Wasserstein Mechanism will always add less noise than \emph{group} differential privacy~\cite{SongYC:17pufferfish}, which makes our findings applicable to group differential privacy as well.
\begin{theorem}\label{theo:wass_noise}
Suppose that $C^H_i$ is the $i$-th biggest connected component of the triggering set, $H$, of a connected graph $G$. Suppose further that $X$ is the total number of activated nodes and $\W(W,\eps)$ is a Wasserstein Mechanism that adds noise to $X$ following $\text{Lap}(W/\eps)$ to provide $\eps$-Pufferfish for a subset $V' \subseteq V$ of nodes. If we have: (1) $|C^H_1| = \Theta(n)$ and (2) $s|C^H_2| = o(n)$, then
$W = \Omega(n)$ for any $V' \subseteq V$.
\end{theorem}
\begin{proof}
For a node $v$, suppose that $\mu^v_i = P(X = \cdot \mid x_v = i)$ for $i \in \{0, 1\}$.
Similar to Theorem~\ref{theo:hiding_c^H_1}, we denote by $\theta_0$ the \emph{largest} cascade when $C^H_1$ is \emph{not} activated, and by $\theta_1$ the \emph{smallest} cascade when $C^H_1$ \emph{is} activated. Note that $\theta_1 - \theta_0 = \Omega(n)$ according to conditions (1) and (2).
Then, the probability of $X \leq \theta_0$ when $v$ is activated is:
\begin{equation}\label{eq:mu_1_<=l}
    \mu_1(X \leq \theta_0) = \Pr[v \notin C^H_1]\Pr[S \cap C^H_1 = \emptyset].
\end{equation}
$C^H_1$ should not receive any seed nodes, and $v$, which is activated should not be in $C^H_1$.
For the probability of $X \leq \theta_0$ when $v$ is \emph{not} activated we have:
\begin{multline}\label{eq:mu_0_<=l}
    \mu_0(X \leq \theta_0) = \Pr[v \in C^H_1] + \\ \Pr[v \notin C^H_1]\Pr[S \cap C^H_1 = \emptyset].
\end{multline}
Let $\mathcal{A} = [0, \theta_0] \times [\theta_1, n]$.
For any joint distribution $\gamma$ with marginals $\mu_0$ and $\mu_1$, we have:
\begin{align}\label{eq:probability_transfer}
    \int_{(x,y)\in \mathcal{A}} \gamma(x,y)d(x,y) &\geq \left| \mu_0(X \leq \theta_0) - \mu_1(X \leq \theta_0) \right| \\
    &\geq \Pr[v \in C^H_1].
\end{align}
\eqref{eq:probability_transfer} is due to the fact that any difference in the total mass of probability in the region $X \leq \theta_0$ between $\mu_0$ and $\mu_1$ has to be transferred into the region $X \geq \theta_1$. Due to \eqref{eq:mu_1_<=l} and \eqref{eq:mu_0_<=l}, this difference is equal to $\Pr[v \in C^H_1]$. Since $G$ is connected, for any $v$: $\Pr[v \in C^H_1] > 0$, which yields:
\begin{equation}
    \exists \, x \in [0, \theta_0], y \in [\theta_1, n] \colon \, \gamma^*(x,y) > 0,
\end{equation}
resulting into:
\begin{equation*}
    W_\infty(\mu_0, \mu_1) \geq \theta_1 - \theta_0 = \Omega(n).
\end{equation*}
Since $W$ is selected by taking the maximum $W_\infty$ across all $v \in V'$, we have $W = \Omega(n)$.
\end{proof}
Note that conditions (1) and (2), as we discussed in Section~\ref{sec:high_degree_privacy}, can be satisfied both in the random networks we studied and in real-world networks.
With $W = \Omega(n)$, for reasonable values of $\varepsilon$, such as a constant, the added noise to protect \emph{any} subset of nodes follows $\text{Lap}(\Omega(n))$ distribution, which destroys data utility.

It is worth noting that the underlying reason for the massive difference between $\mu^v_1$ and $\mu^v_0$ for any node $v$ is indeed the disparity between the values of $X$ when $C^H_1$ is activated and when it is not.
\section{Experiments}\label{sec:exp}

We ran experiments on real world data sets to show that: (1) Giant connected components appear in cascades over real networks under reasonable values of $q$, the transmission rate, and (2) a sizeable number of nodes in these networks appear frequently in the giant connected component, and as a result of the difficulty in hiding the giant connected component's activation status, they are highly vulnerable to inference attacks.

We conduct our experiments on $4$ real-world networks (Table~\ref{tab:real_world}): two co-authorship networks, GrQc and HepTh~\cite{co-authorship} and two co-purchase networks, Amazon DVDs and Amazon Musics~\cite{amazon-co-purchase}. The last two columns of this table includes the average size of the giant component ($|\overline{C^H_1}|$) and the second largest component ($|\overline{C^H_2}|$) of the triggering set, $H$, when $q$ is $0.3$.
\begin{table}[!tbh]
    \centering
    \caption{Real-world networks.}
    \begin{tabular}{lcccc}
        \toprule
        Network & \#Nodes & \#Edges & $|\overline{C^H_1}|$ & $|\overline{C^H_2}|$ \\
        \midrule
        GrQc & $2.4\text{k}$ & $10.9\text{k}$ & $1{,}577.3$ & $28.99$ \\
        HepTh & $4.9\text{k}$ & $19.4\text{k}$ & $3{,}494.7$ & $19.73$ \\
        Amz. Music & $29.4\text{k}$ & $84.3\text{k}$ & $9{,}804.0$ & $197.23$ \\
        Amz. DVDs & $9.5\text{k}$ & $27.6\text{k}$ & $4{,}341.5$ & $139.39$\\
        \bottomrule
    \end{tabular}
    \label{tab:real_world}
\end{table}
\begin{figure}[!tb]
    \centering
    \includegraphics[width=0.85\columnwidth]{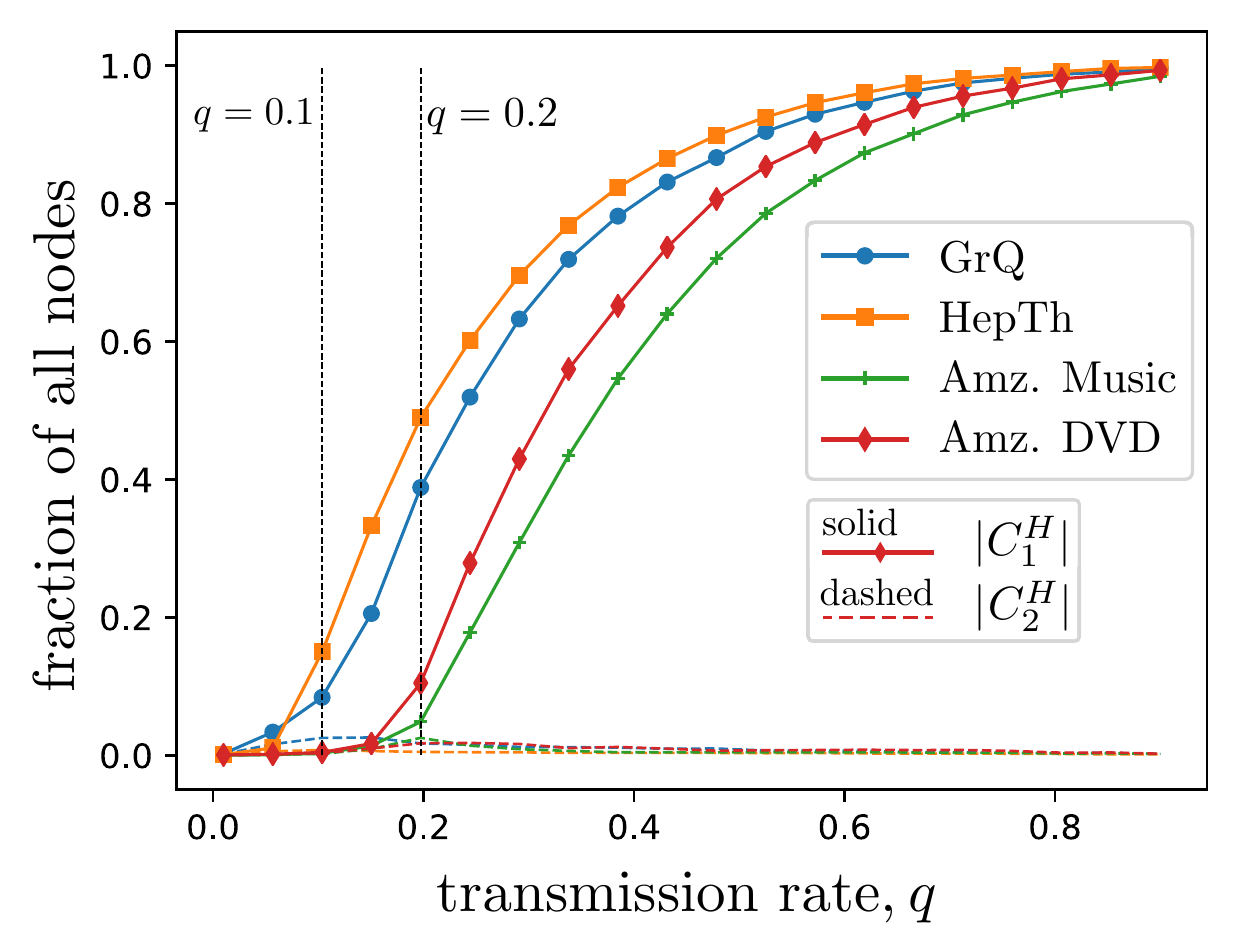}
    \caption{The average size of the largest, $C^H_1$ in solid marked lines, and the second largest connected components, $C^H_2$ in dotted lines, of the triggering set as fraction of all nodes, with with different values of $q$.}
    \label{fig:real_phase_transition}
    \vspace*{-7mm}
\end{figure}
\begin{table*}[!tbh]
    \centering
    \caption{Number and percentage of nodes with probability of being in $C^H_1$ higher than different thresholds.}
    \begin{tabular}{lccccc}
        \toprule
        \multirow{2}[2]{*}{Network} & \multicolumn{5}{c}{Threshold for $\Pr[v \in C^H_1]$} \\
        \cmidrule(lr){2-6}
        &  $\geq 99\%$ & $\geq 95\%$ & $\geq 90\%$ & $\geq 75\%$ & $\geq 50\%$\\
        \midrule
        GrQc & $211 \; (8.71\%)$ & $383 \; (15.81\%)$ & $532 \; (21.97\%)$ & $906 \; (37.41\%)$ & $1{,}765 \; (72.87\%)$ \\
        HepTh & $462 \; (9.41\%)$ & $966 \; (19.68\%)$ & $1371$ $(27.93\%)$ & $2414$ $(49.17\%)$ & $4054$ $(82.58\%)$ \\
        Amz. Music & $153$ $(0.52\%)$ & $408$ $(1.39\%)$ & $751$ $(2.56\%)$ & $2865$ $(9.76\%)$ & $8506$ $(28.98\%)$\\
        Amz. DVDs & $58$ $(0.61\%)$ & $203$ $(2.14\%)$ & $414$ $(4.36\%)$ & $1647$ $(17.36\%)$ & $4529$ $(47.73\%)$\\
        \midrule
        ER($2500$,$5$) & $1$ $(0.0\%)$ & $3$ $(0.1\%)$ & $10$ $(0.4\%)$ & $433$ $(17.3\%)$ & $1{,}737$ $(69.5\%)$\\
        CL($2500$,$5$,$1.1$) & $668$ $(26.7\%)$ & $994$ $(39.8\%)$ & $1{,}269$ $(50.8\%)$ & $2{,}043$ $(81.7\%)$ & $2{,}459$ $(98.4\%)$\\
        CL($2500$,$5$,$1.5$) & $952$ $(38.1\%)$ & $1{,}534$ $(61.4\%)$ & $1{,}902$ $(76.1\%)$ & $2{,}354$ $(94.2\%)$ & $2{,}481$ $(99.2\%)$\\
        CL($2500$,$5$,$2$) & $622$ $(24.9\%)$ & $1{,}171$ $(46.8\%)$ & $1{,}599$ $(64.0\%)$ & $2{,}234$ $(89.4\%)$ & $2{,}446$ $(97.8\%)$\\
        \bottomrule
    \end{tabular}
    \label{tab:P(v_in_giant)}
\vspace{-4mm}
\end{table*}
For each network, $1000$ triggering sets are generated at random.
As a result of the massive difference seen between $|C^H_1|$ and $|C^H_2|$, it will be extremely hard to hide the activation status of $C^H_1$. To find out for what values of $q$ a giant connected component  appears, we have tested $20$ values of $q$ between $0.05$ and $0.9$ and reported $\overline{|C^H_1|}$ and $\overline{|C^H_2|}$ over $50$ realizations. The results are depicted in Figure~\ref{fig:real_phase_transition}. Initially, for extremely low values of $q$, the two are almost indistinguishable in size, but as $q$ grows, $C^H_1$ quickly transitions into the \emph{giant component}, containing a sizeable portion of all nodes while $|C^H_2|$ remains negligible in comparison to $n$.
For values of $q$ as low as $0.1$ in the case of the two co-authorship networks, we already see a significant difference between the two components. This sharp increase happens in the co-purchase networks, Amazon DVD and Music, at around $q = 0.2$. 
In comparison, the transmission probability of high-risk contacts in COVID-19 is estimated at about $15.6\%$~\cite{covid19highriskprob}. 
This shows that in real-world contagions there could be a giant connected component that endangers the privacy of those individuals in it.
Next, we estimate the probability of each node's presence in $C^H_1$ to measure their vulnerability against inference attacks.
The results for the $4$ real-world networks, along with $4$ random synthetic networks for comparison, are available in Table~\ref{tab:P(v_in_giant)}. In each column, we report the number and percentage of nodes whose probability of appearing in $C^H_1$ is above a threshold mentioned at the top. Results are aggregated over $1000$ realizations of $H$ for each network, using $q = 0.3$, and the thresholds are $99\%$, $95\%$, $90\%$, $75\%$ and $50\%$ respectively. First, notice that there are hundreds of nodes across all $4$ networks with values above even very high thresholds, such as $99\%$ or $95\%$.
This means that if the privacy mechanism has Sublinear Asymptotic Error, an adversary can infer the real status of these nodes with a confidence above $95\%$ or $99\%$. In addition, the numbers vary between the co-authorship networks, GrQ and HepTh, and the Amazon DVD and Music networks. As discussed earlier, the structure of each network plays a key role in these probability values.
For comparison, we repeat this experiment on $4$ synthetic randomized networks with $2500$ nodes as well: an \er~graph with $n = 2500$ and $np = 5$, and $3$ Chung-Lu power-law networks with minimum expected degree, $d = 5$ and the scale parameter, $b$ equal to $1.1$, $1.5$, and $2$ (real-world values lie in $(1, 2]$~\cite{barabasi1999emergence}).
With the selected setting, all $4$ of these networks will have a giant component in their corresponding triggering sets w.h.p.
As expected, we observe that the real-world networks are much more resembled by the Chung-Lu power-law networks, rather than the \er~graphs. As we discussed in Section~\ref{sec:power-law}, this is due to the fact that \er~graphs do not model the uneven degree distribution of real-world networks, which is partially responsible for the unbalanced distribution of $\Pr[v \in C^H_1]$ among nodes.

\section{Conclusion and Future Work}\label{sec:conclusion}
It is known that if a contagion follows the Independent Cascade model,
a giant connected component with a constant fraction of all individuals likely appears.  We established that if a privacy mechanism provides sublinear asymptotic error,
hiding the status of this component is impossible. Because of this, many individuals are vulnerable to inference attacks, with those that greatly influence others, a.k.a influencers, in the most danger.
Finally, we have shown that the current state-of-the-art mechanism is incapable of providing sublinear ($o(n)$) privacy for \emph{any} node in the two classes of random networks we studied and in general graphs. 
Our findings put new emphasis on the importance of protecting the topology of social networks.
It is interesting to study other contagion models, such as the Linear Threshold model, seeking for a similar impossibility result.
The final question is, could there be adequate privacy for individuals with new privacy frameworks or minimal changes in our assumptions?
\bibliographystyle{plainnat}
\bibliography{reference,contagion-privacy,jiepub,privacy}
\clearpage
\newpage
\appendix
\section{Pufferfish Privacy}\label{app:pufferfish}
In Pufferfish privacy, the model seeks to protect pairs of ``secrets'' by making the two secrets \emph{indistinguishable} from each other by an adversary given the output of the privacy mechanism, $\M(X)$.
In our case, the two secrets are whether $\x_v = 0$ or $\x_v = 1$ for each $v$. Pufferfish privacy also provide guarantees with respect to a set of \emph{data-generating} distributions governing how the dataset is created. The meaning of the privacy guarantee is then conditional: if the data was generated according to the assumed family of data-generating distributions, then a privacy mechanism will make it hard for an adversary to distinguish whether secret $i$, e.g. $\x_v = 0$, or secret $j$, e.g. $\x_v = 1$, is true given that $\M(X)$ is known by the public. Mechanisms that guarantee Pufferfish privacy introduce noise into the computation to provide this guarantee.

More formally,
let $\mathbb{S}$ be a set of secrets, $\mathbb{Q} \subseteq \mathbb{S} \times \mathbb{S}$ a set of pairs of secrets we wish an adversary cannot distinguish between, and $\Pi$ a set of \emph{data-generating} distributions, $\pi \in \Pi$, governing how the dataset is created.
The formal definition of Pufferfish privacy is given below.
\begin{definition}[Pufferfish Privacy~\cite{kifer2014pufferfish}]\label{def:pufferfish} A randomized mechanism $\M$ in framework ($\mathbb{S},\mathbb{Q},\Pi$) provides $\varepsilon$-Pufferfish privacy if for any data-generating distribution $\pi \in \Pi$, any pairs of secrets $s_i, s_j \in \mathbb{Q}$ such that $P(s_i \mid \pi), P(s_j \mid \pi) > 0$, and for all possible outputs $w \in \text{Range}(\M)$ we have:
\begin{equation}\label{eq:pufferfish}
    \frac{
        \Pr[\M(X) = w \mid s_i, \pi]}{
        \Pr[\M(X) = w \mid s_j, \pi]} \leq e^\varepsilon.
\end{equation}
\end{definition}
Note that the above is almost exactly the guarantee seen in Differential Privacy. The key difference is that the randomness of $X$, the dataset, is taken into consideration when the probability ratio in \eqref{eq:pufferfish} is calculated.
The goal of Pufferfish privacy is to incorporate correlations between data points, via the assumptions on $\Pi$, to enable us to measure the privacy guarantees for correlated data; something that DP is known to lack. For instance, in our case the set of data-generating distributions could be \emph{all} possible graphs generated by a random network given fixed parameters. Alternatively, with a fixed graph, the set $\Pi$ could include all possible triggering sets. The flexibility of $\mathbb{S}$ and $\mathbb{Q}$ allows the mechanism designer to only aim for the protection of the privacy of a subset of nodes, not all of them.

\section{Proofs of Lemmas~\ref{lem:P(x_v|M(X))_bound} and \ref{lem:trig_approx}}\label{app:trig_approx_proof}
\subsection{Lemma~\ref{lem:P(x_v|M(X))_bound}}
\begin{proof}
As we discussed in Section~\ref{sec:er_graphs}, the condition in the first equation means that $C^H_1$ is activated. Once $C^H_1$ is active, one way for node $i$ to become activated is to be a part of $C^H_1$. Obviously, there are other ways for that to happen, hence the lower bound. An exactly similar argument goes for the second equation as well.
\end{proof}
\subsection{Lemma~\ref{lem:trig_approx}}
\begin{proof}
Since we are fixing the degree to be $k$, and each node has an equal probability of being one of the $k$ neighbors of $i$, the probability that $i$ chooses \emph{none} of the nodes in $C^H_1$ as a neighbor is equal to:
\begin{equation}\label{eq:choose_exponential}
    \frac{\binom{n - m}{k}}{\binom{n}{k}} \sim - \exp\left(\frac{km}{n}\right).
\end{equation}
In the above, we have used Stirling's approximation of $n! \sim \sqrt{2\pi}(n/e)^n$.
Substituting $m$ in \eqref{eq:choose_exponential} with $yn \pm o(n)$ as the number of nodes in  $C^H_1$ yields \eqref{eq:i_C^H_1}.
\end{proof}
\section{Proof of Theorem~\ref{theo:cl_influencer_privacy}}\label{app:cl_influence_proof}
\begin{proof}
First notice that unlike Theorem~\ref{theo:er_influencer_privacy}, here we also prove that the subset of nodes that will not have adequate protection by $\M$ can in fact grow arbitrarily in size.
In other words:
\[
\lim_{n \rightarrow \infty} \left| V' \right| \rightarrow \infty.
\]
Similar to the arguments in Theorem~\ref{theo:er_influencer_privacy}, we can observe that either $C^H_1$ is active and $\M(X) \geq |C^H_1| - e_{\M} = \Omega(n)$, or it is inactive and $\M(X) \leq s|C^H_2| + e_{\M} = o(n)$. Because of this, one can infer the status of $C^H_1$ given the value of $\M(X)$.
We can now use Lemma~\ref{lem:P(x_v|M(X))_bound} here to bound the two probabilities in \eqref{eq:cl_influencer_privacy} by $\Pr[i \in C^H_1]$. For its complement, $\Pr[i \notin C^H_1]$ we have:
\begin{align}
    \Pr[i \notin C^H_1] &\leq \prod_{j \in C^H_1} 1 -  \min\left(1, \frac{w_i w_j}{\ell_n}\right) \nonumber \\
    & \leq \left(1 -  \min\left(1, \frac{w_i w_n}{\ell_n}\right)\right)^{\alpha n}, \label{eq:first_step_cl_lb}
\end{align}
since $w_n = dq < w_j$ $(n > j)$ due to \eqref{eq:w_i}. In extreme cases when $b$ is too small, we might have $w_1 w_n / \ell_n \geq 1$, making the node with the highest degree ($i = 1$) connected to every single node with probability $1$ and the whole graph becomes connected. This means that
with any number of seed nodes $s > 0$, the \emph{whole} network is activated and there could be no privacy for any node. Apart from such cases, which are rare in reality as $b$ usually lies in $(1,2]$ in real-world networks~\cite{barabasi1999emergence}, we can assume that $w_i w_n \leq \ell_n$ and rewrite \eqref{eq:first_step_cl_lb} as:
\begin{align}
    \Pr[i \notin C^H_1] & \leq \exp\left(-\frac{w_i w_n \alpha n}{\ell_n}\right) \nonumber\\
    & \leq \exp\left(-\frac{d^2q^2 \alpha n^{\beta + 1}}{i^\beta dqn^\beta\sum_j j^{-\beta}}\right) \nonumber \\
    & \leq \exp\left(-\frac{dq \alpha n}{i^\beta \sum_j j^{-\beta}}\right). \label{eq:second_step_cl_lb}
\end{align}
The above approaches $0$ as $n$ grows large enough when:
\begin{equation}\label{eq:general_cl_bound}
    i = o\left(\left[\frac{n}{\sum_j j^{-\beta}}\right]^{1/\beta}\right)
\end{equation}
Depending on the value of $\beta$ ($1/b$), there could be $3$ cases for
\eqref{eq:general_cl_bound}:
\benum
\item In case $\beta > 1$ ($0 < b < 1$), the summation $\sum^\infty_j j^{-\beta}$ is a constant number. Then, for $i$ we can write:
\begin{equation}\label{eq:cl_bound_1/b>0}
    i = o\left(n^{1/\beta}\right) = o\left(n^b\right).
\end{equation}
\item In case $\beta = 1$, the summation $\sum^n_{j=0} j^{-1}$ tends to $\Theta(\log{n})$. Using this, we can write:
\begin{equation}\label{eq:cl_bound_b=1}
i = o\left(\frac{n}{\log{n}}\right)
\end{equation}
\item In case $0 < \beta < 1$, the partial summation $\sum_{j=0}^n j^{-\beta}$ is $\Theta(n^{1 - \beta})$~\cite{p-series-bound}. Using this, we can write:
\begin{equation}\label{eq:cl_bound_1/b<0}
i = o\left(\left[\frac{n}{n^{1 - \beta}}\right]^{1/\beta}\right) = o(n)
\end{equation}
\eenum
The $3$ equations above show that
for any $\epsilon$ in Theorem~\ref{theo:cl_influencer_privacy}, the subset $V' \subseteq V$ includes all nodes ranked $i$ where $i \leq f(n)$ and $f$ meets the corresponding asymptotic condition depending on $b$ in the equations above. Note that $f(n)$, and subsequently $|V'|$, can grow arbitrarily large. For instance in the case of $b = 1$, to meet the condition in \eqref{eq:cl_bound_b=1}, we can set $f(n) = n/\log^2{n}$. This concludes the proof.
\end{proof}

\end{document}